\documentclass[11pt]{article}
\usepackage{hyperref}
\usepackage{times}  
\usepackage{mathpazo}
\usepackage{amssymb,amsmath,amsthm}
\usepackage{epsfig}
\usepackage{tcolorbox}
\usepackage{lipsum} 
\usepackage{enumitem}
\usepackage{xfrac}

\tcbset{ 
  halign=justify, 
  center, 
  colback=gray!10, 
  colframe=black, 
  boxrule=0.5pt 
}

\newtheorem{theorem}{Theorem}[section]

\newtheorem{definition}[theorem]{Definition}

\newtheorem{lemma}[theorem]{Lemma}

\newtheorem{remark}[theorem]{Remark}

\newcommand{\qedsymb}{\hfill{\rule{2mm}{2mm}}}
\renewenvironment{proof}[1][]{\begin{trivlist}
\item[\hspace{\labelsep}{\bf\noindent Proof#1:\/}] }{\qedsymb\end{trivlist}}

\def\calR{{\cal R}}

\def\calF{{\cal F}}

\def\calJ{{\cal J}}
\def\calE{{\cal E}}
\def\calD{{\cal D}}
\def\calM{{\cal M}}
\def\calP{{\cal P}}

\def\calH{{\cal H}}

\newcommand\Prob[2]{{\Pr_{#1}\left[ {#2} \right]}}

\newcommand{\eps}{\epsilon}
\renewcommand{\epsilon}{\varepsilon}

\newcommand{\poly}{\mathop{\mathrm{poly}}}

\newcommand{\Inter}{\textsc{Intersecting}}
\newcommand{\Sampler}{\textsc{Canonical Tester}}

\begin{document}

\title{{\bf Testing Intersectingness of Uniform Families}}

\author{
Ishay Haviv\thanks{School of Computer Science, The Academic College of Tel Aviv-Yaffo, Tel Aviv 61083, Israel. Research supported in part by the Israel Science Foundation (grant No.~1218/20).}
\and
Michal Parnas\thanks{School of Computer Science, The Academic College of Tel Aviv-Yaffo, Tel Aviv 61083, Israel.}
}

\date{}

\maketitle

\begin{abstract}
A set family $\calF$ is called intersecting if every two members of $\calF$ intersect, and it is called uniform if all members of $\calF$ share a common size.
A uniform family $\calF \subseteq \binom{[n]}{k}$ of $k$-subsets of $[n]$ is $\eps$-far from intersecting if one has to remove more than $\eps \cdot \binom{n}{k}$ of the sets of $\calF$ to make it intersecting.
We study the property testing problem that given query access to a uniform family $\calF \subseteq \binom{[n]}{k}$, asks to distinguish between the case that $\calF$ is intersecting and the case that it is $\eps$-far from intersecting.
We prove that for every fixed integer $r$, the problem admits a non-adaptive two-sided error tester with query complexity $O(\frac{\ln n}{\eps})$ for $\eps \geq \Omega( (\frac{k}{n})^r)$ and a non-adaptive one-sided error tester with query complexity $O(\frac{\ln k}{\eps})$ for $\eps \geq \Omega( (\frac{k^2}{n})^r)$.
The query complexities are optimal up to the logarithmic terms.
For $\eps \geq \Omega( (\frac{k^2}{n})^2)$, we further provide a non-adaptive one-sided error tester with optimal query complexity of $O(\frac{1}{\eps})$.
Our findings show that the query complexity of the problem behaves differently from that of testing intersectingness of non-uniform families, studied recently by Chen, De, Li, Nadimpalli, and Servedio (ITCS,~2024).
\end{abstract}

\section{Introduction}

A set family $\calF$ is called intersecting if for every two sets $F_1, F_2 \in \calF$, it holds that $F_1 \cap F_2 \neq \emptyset$.
The study of intersecting families plays a central role in the area of extremal combinatorics with a particular attention dedicated to the uniform case, where all the sets of the family share a common size.
One of the most influential results in this context is the Erd{\H{o}}s--Ko--Rado theorem~\cite{EKR61}, proved in 1938 and published in 1961, which states that for integers $n$ and $k$ with $n \geq 2k$, the maximum size of an intersecting family of $k$-subsets of $[n] = \{1,2,\ldots,n\}$ is $\binom{n-1}{k-1}$, attained by the families of all $k$-subsets that include a fixed element.
Another prominent result, proved by Lov{\'{a}}sz~\cite{LovaszKneser} in 1978 settling a conjecture of Kneser~\cite{Kneser55} from 1955, asserts that for $n \geq 2k$, the family $\binom{[n]}{k}$ of all $k$-subsets of $[n]$ cannot be covered by fewer than $n-2k+2$ intersecting families.
This result is tight, as follows by considering, for each $i \in [n-2k+1]$, the family of $k$-subsets of $[n]$ that include $i$, and the family of $k$-subsets of $[n] \setminus [n-2k+1]$.
A more recent result, proved by Dinur and Friedgut~\cite{DinurF09} in 2009, provides a structural characterization for large intersecting families of $k$-subsets of $[n]$ when $k$ is sufficiently smaller than $n$. It says, roughly speaking, that every such family is approximately contained in an intersecting junta, that is, an intersecting family $\calJ$ over $[n]$, such that the membership of a set $F$ in $\calJ$ depends only on $F \cap J$ for a fixed set $J \subseteq [n]$, where the size of $J$ is determined by the precision of the containment (see Theorem~\ref{thm:DF}).

In this paper, we investigate intersecting uniform families from the computational perspective of property testing.
This field delves into the amount of data required for distinguishing objects that satisfy a prescribed property from those that significantly deviate from satisfying it.
The objective is to design a randomized algorithm for this task, called a (two-sided error) tester, that succeeds with high constant probability and minimizes the query complexity, i.e., the number of queries to the input object.
If the tester accepts objects that satisfy the given property with probability $1$, we say that its error is one-sided.
The tester is said to be non-adaptive if its queries are determined independently of the answers provided for prior queries.
For a thorough introduction to the field of property testing, the reader is referred to, e.g.,~\cite{Goldreich17}.

For integers $n$ and $k$ with $n \geq 2k$ and for a real $\eps \in [0,1)$, we say that a family $\calF \subseteq \binom{[n]}{k}$ is $\eps$-far from intersecting if one has to remove more than $\eps \cdot \binom{n}{k}$ of its sets to make it intersecting. In the property testing problem $\Inter_{n,k,\eps}$, we are given access to a family $\calF \subseteq \binom{[n]}{k}$, represented by an indicator function $f: \binom{[n]}{k} \rightarrow \{0,1\}$, and the goal is to distinguish between the case that $\calF$ is intersecting and the case that it is $\eps$-far from intersecting.
Note that the Erd{\H{o}}s--Ko--Rado theorem~\cite{EKR61} implies that every intersecting family $\calF \subseteq \binom{[n]}{k}$ includes at most a $k/n$ fraction of the sets of $\binom{[n]}{k}$. This gives impetus to studying the $\Inter_{n,k,\eps}$ problem for a proximity parameter $\eps = \eps(n,k)$ that diminishes faster than the ratio $k/n$ (see the discussion at the end of Section~\ref{sec:two-sided}).

Our interest in testing intersectingness of uniform families is sparked and inspired by a recent paper of Chen, De, Li, Nadimpalli, and Servedio~\cite{CDLNS24}, who introduced and explored the analogue problem for general (non-uniform) families of sets.
In their setting, the input consists of an indicator function $f: \{0,1\}^n \rightarrow \{0,1\}$ of a family $\calF$ of subsets of $[n]$ (of any size), identified with their characteristic vectors in $\{0,1\}^n$, and the goal is to decide whether $\calF$ is intersecting or $\eps$-far from intersecting. Here, since the size of the domain is $2^n$, a family is said to be $\eps$-far from intersecting if more than $\eps \cdot 2^n$ of its sets should be removed to make it intersecting.
Chen et al.~\cite{CDLNS24} proved that this problem admits a non-adaptive one-sided error tester with query complexity $\poly(n^{\sqrt{n \cdot \log(1/\eps)}},\frac{1}{\eps})$.
They further established a nearly matching lower bound, showing that the query complexity of every non-adaptive one-sided error tester for the problem is $2^{\Omega(\sqrt{n \cdot \log (1/\eps)})}$, whenever $\eps \in [2^{-n},\eps_0]$ for some constant $\eps_0>0$.
For non-adaptive two-sided error testers, they obtained a lower bound of $2^{\Omega(n^{1/4}/\sqrt{\eps})}$ on the query complexity, assuming that $\eps \in [1/\sqrt{n},\eps_0]$ for some constant $\eps_0>0$.
As will be shortly described, the results of the present paper highlight a significant difference between the behavior of the query complexity of testing intersectingness in the uniform and non-uniform settings.

\subsection{Our Contribution}

This paper studies the query complexity of the $\Inter_{n,k,\eps}$ testing problem.
We offer nearly matching upper and lower bounds for various settings of the integers $n$ and $k$ and the proximity parameter $\eps = \eps(n,k)$.
Let us mention already here that all of our upper bounds are proved via efficient testers, whose running time is polynomial in $n$.
We further note that the results presented below are applicable for all integers $n$ and $k$ for which the conditions on $\eps$ allow it to be smaller than $1$.
The precise and formal statements are given in the upcoming technical sections.
A summary of our main upper bounds is given in Table~\ref{tab:mytable}.

Our first result furnishes a non-adaptive two-sided error tester for the case of $\eps \geq \Omega( (k/n)^r)$ for a fixed constant $r$.
Its analysis borrows the structural characterization of large intersecting uniform families due to Dinur and Friedgut~\cite{DinurF09}.
Note that the multiplicative factors hidden by the $O(\cdot)$ and $\Omega(\cdot)$ notations might depend on $r$.

\begin{theorem}\label{thm:1Intro}(Two-Sided Error Tester)
For every fixed integer $r$, for all integers $n$ and $k$ with $n \geq 2 k$ and for any real $\eps \geq \Omega( (\frac{k}{n})^r)$, there exists a non-adaptive two-sided error tester for $\Inter_{n,k,\eps}$ with $O(\frac{\ln n}{\eps})$ queries.
\end{theorem}

In fact, we prove Theorem~\ref{thm:1Intro} in a stronger form, adopting the concept of tolerant property testing, introduced by Parnas, Ron, and Rubinfeld~\cite{ParnasRR06}.
Strengthening the standard notion of property testing, a tolerant tester is required to accept not only objects that satisfy the given property, but also those that are close to satisfying it (and as usual, to reject objects that significantly deviate from the property).
Accordingly, the tester given in Theorem~\ref{thm:1Intro} is shown to accept with high probability any family that can be made intersecting by removing relatively few of its sets (see Theorem~\ref{thm:two-sided}). Note that a result of Tell~\cite{Tell20} implies that a two-sided error is essentially inherent in tolerant property testing.

We next turn our attention to designing one-sided error testers for the $\Inter_{n,k,\eps}$ problem, wherein an intersecting family must be accepted with probability $1$. A natural non-adaptive tester for this purpose, termed the canonical tester, selects $m$ random sets from $\binom{[n]}{k}$, uniformly and independently, and checks whether they include two sets that demonstrate the non-intersectingness of the input family, namely, two disjoint sets within the family. This raises the combinatorial question, which might be of independent interest, of determining the smallest number $m=m(n,k,\eps)$ of random sets from $\binom{[n]}{k}$ that guarantee with high probability a pair of disjoint sets that lie in a family $\calF \subseteq \binom{[n]}{k}$, assuming that $\calF$ is $\eps$-far from intersecting. As our main technical contribution, we address this question for the case where $\eps \geq \Omega( (k^2/n)^r)$ for a fixed constant $r$. Our analysis yields the following result.

\begin{theorem}\label{thm:2Intro}(One-Sided Error Tester)
For every fixed integer $r$, for all integers $n$ and $k$ with $n \geq 2 k$ and for any real $\eps \geq \Omega( (\frac{k^2}{n})^r)$, there exists a non-adaptive one-sided error tester for $\Inter_{n,k,\eps}$ with $O(\frac{\ln k}{\eps})$ queries.
\end{theorem}
\noindent
Let us emphasize that the tester provided by Theorem~\ref{thm:2Intro} surpasses that of Theorem~\ref{thm:1Intro} in two respects: its error is one-sided, and its query complexity is lower, replacing the $\ln n$ term by $\ln k$.
On the other hand, Theorem~\ref{thm:2Intro} requires $\eps$ to satisfy $\eps \geq \Omega( (k^2/n)^r)$, and is thus applicable only for $n \geq \Omega(k^2)$, whereas the two-sided error tester of Theorem~\ref{thm:1Intro} is applicable already for $n \geq \Omega(k)$.

For the special case of $r=2$, we offer a notably simple analysis of the canonical tester, enabling us to enhance the query complexity achieved in Theorem~\ref{thm:2Intro} by getting rid of the logarithmic term.
This gives the following result.

\begin{theorem}\label{thm:3Intro}(One-Sided Error Tester for $r=2$)
For all integers $n$ and $k$ with $n \geq 2 k$ and for any real $\eps \geq \Omega( (\frac{k^2}{n})^2)$, there exists a non-adaptive one-sided error tester for $\Inter_{n,k,\eps}$ with $O(\frac{1}{\eps})$ queries.
\end{theorem}

We further consider the $\Inter_{n,k,\eps}$ problem for integers $n$ and $k$ satisfying $n = \alpha \cdot k$ for an arbitrary constant $\alpha \geq 2$.
Interestingly, the canonical tester fails in this case, because its random samples are unlikely to include even a single pair of disjoint sets, unless the number of samples is exponential in $n$. Nevertheless, we show that for all constants $\alpha \geq 2$ and $\eps \in (0,1)$, the $\Inter_{n,k,\eps}$ problem with $n = \alpha \cdot k$ admits a non-adaptive one-sided error tester with constant query complexity.
The proof employs a result of Friedgut and Regev~\cite{FR18}, and the details are given in Section~\ref{sec:n=alpha*k} (see Theorem~\ref{thm:n=alpha*k}).

\begin{table}[ht]
    \centering
    \begin{tabular}{cccc}
        Tester & Condition & Query complexity & Result\\
        \hline \hline
        Two-sided error, tolerant & $\eps \geq \Omega( \frac{k}{n})$ & $O(\frac{1}{\eps})$ & Section~\ref{sec:two-sided}\\[1ex]
        & $\eps \geq \Omega( (\frac{k}{n})^r),~\mbox{fixed}~r \geq 2$ & $O(\frac{\ln n}{\eps})$ & Theorems~\ref{thm:1Intro},~\ref{thm:two-sided}\\[1ex]
        One-sided error & $\eps \geq \Omega( (\frac{k^2}{n})^2)$ & $O(\frac{1}{\eps})$ & Theorem~\ref{thm:3Intro}\\[1ex]
        & $\eps \geq \Omega( (\frac{k^2}{n})^r),~\mbox{fixed}~r \geq 3$ & $O(\frac{\ln k}{\eps})$ & Theorem~\ref{thm:2Intro}\\[1ex]
        
    \end{tabular}
    \caption{The query complexity of the $\Inter_{n,k,\eps}$ problem.}
    \label{tab:mytable}
\end{table}

Our final result supplies a lower bound on the query complexity of the $\Inter_{n,k,\eps}$ problem.
Note that the lower bound applies even to adaptive two-sided error testers.
While the result can be derived from a general result of~\cite{Fischer24}, we give a full proof which explicitly presents hard instances of the problem.

\begin{theorem}\label{thm:LBIntro}(Lower Bound)
For all integers $n$ and $k$ with $n \geq 2k$ and for any real $\eps = \eps(n,k)$ with $\binom{n}{k}^{-1} \leq \eps < \frac{1}{2}$, the query complexity of every tester for $\Inter_{n,k,\eps}$ is $\Omega(\frac{1}{\eps})$.
\end{theorem}

It is noteworthy that Theorem~\ref{thm:LBIntro} implies that the query complexities achieved by our testers for the $\Inter_{n,k,\eps}$ problem are nearly tight.
Specifically, the query complexity obtained in Theorem~\ref{thm:3Intro} is tight up to a multiplicative constant, while those of Theorems~\ref{thm:1Intro} and~\ref{thm:2Intro} are tight up to multiplicative logarithmic terms.
An intriguing task for further research would be to decide whether these logarithmic terms can be avoided.
More ambitiously, it would be interesting to determine the query complexity of the $\Inter_{n,k,\eps}$ problem for general values of $n$, $k$, and $\eps = \eps(n,k)$.

\subsection{Proof Techniques}

We provide here a high-level description of the ideas applied in the proofs of Theorems~\ref{thm:1Intro} and~\ref{thm:2Intro}.
Let us start with the proof of Theorem~\ref{thm:1Intro}, which gives a non-adaptive two-sided error tester for the $\Inter_{n,k,\eps}$ problem with query complexity $O(\frac{\ln n}{\eps})$, where $\eps \geq \Omega( (\frac{k}{n})^r)$ for a fixed integer $r$. Given access to a family $\calF \subseteq \binom{[n]}{k}$, our tester attempts to decide whether $\calF$ is approximately contained in an intersecting $j$-junta over $[n]$ for some integer $j$ that depends solely on $r$. To do so, the tester picks random sets from $\binom{[n]}{k}$, uniformly and independently, and checks whether they lie in $\calF$. These samples are used to estimate, for each intersecting $j$-junta $\calJ$ over $[n]$, the fraction of sets in $\binom{[n]}{k}$ that lie in $\calF \setminus \calJ$. Our analysis shows that $O(\frac{\ln n}{\eps})$ samples suffice for these estimations to be pretty accurate with high probability. Then, if those estimations indicate that $\calF$ is approximately contained in some intersecting $j$-junta over $[n]$, the tester predicts that $\calF$ is close to intersecting. Otherwise, relying on the aforementioned structural result of Dinur and Friedgut~\cite{DinurF09}, the tester deduces that $\calF$ is far from intersecting with high probability.
Let us stress that the tester is two-sided error, because even if $\calF$ is intersecting and is essentially aligned with some intersecting $j$-junta $\calJ$ over $[n]$, the random samples might wrongly suggest, with low probability, that $\calF$ significantly deviates from $\calJ$.
As previously noted, Theorem~\ref{thm:1Intro} is proved with respect to tolerant property testing.
For the precise statement and argument, the reader is referred to Section~\ref{sec:two-sided}.

We next outline the approach applied in the proof of Theorem~\ref{thm:2Intro}, which establishes a non-adaptive one-sided error tester for the $\Inter_{n,k,\eps}$ problem with query complexity $O(\frac{\ln k}{\eps})$, where $\eps \geq \Omega( (\frac{k^2}{n})^r)$ for a fixed integer $r$. As mentioned earlier, the proof of this result relies on the canonical tester, which selects random sets from $\binom{[n]}{k}$, uniformly and independently, and checks if they include two sets that violate the intersectingness of the given family. Therefore, our goal is to show that if a family $\calF \subseteq \binom{[n]}{k}$ is $\eps$-far from intersecting, then a collection of $O(\frac{\ln k}{\eps})$ random sets from $\binom{[n]}{k}$ includes with high probability two disjoint sets within $\calF$.
Consider a family $\calF \subseteq \binom{[n]}{k}$, and assume that $\calF$ is $\eps$-far from intersecting. Notice that this in particular implies that $|\calF| > \eps \cdot \binom{n}{k}$, as otherwise, one could remove at most $\eps \cdot \binom{n}{k}$ of its sets to make it intersecting.

Let us first suppose that the chosen random sets include, for each set $A \subseteq [n]$ of size at most $r-1$, a set $F_A \in \calF$ with $F_A \cap A = \emptyset$. A key observation in our approach is that the number of sets in $\binom{[n]}{k}$ that intersect all of those sets $F_A$ does not exceed $k^r \cdot \binom{n-r}{k-r}$. To see this, observe that every set that intersects them all includes an element $j_1 \in F_\emptyset$, an element $j_2 \in F_{\{j_1\}}$, an element $j_3 \in F_{\{j_1, j_2\}}$, and so on, up to an element $j_r \in F_{\{j_1, \ldots, j_{r-1}\}}$, where the $r$ elements $j_1, \ldots, j_r$ are distinct. It thus follows that there exists a collection of at most $k^r$ subsets of $[n]$ of size $r$, such that every set in $\binom{[n]}{k}$ that intersects all the sampled sets that lie in $\calF$ contains at least one of the sets of the collection. This implies that the number of those sets does not exceed $k^r \cdot \binom{n-r}{k-r} \leq (\frac{k^2}{n})^r \cdot \binom{n}{k}$. Now, by combining our assumption on $\eps$ with the fact that $|\calF| > \eps \cdot \binom{n}{k}$, one can show that a random set from $\binom{[n]}{k}$ lies in $\calF$ and is disjoint from at least one of the previously sampled sets that lie in $\calF$ with non-negligible probability. Therefore, a few additional random sets from $\binom{[n]}{k}$ are expected to provide with high probability the desired witness for the non-intersectingness of $\calF$.

The scenario discussed above, however, is somewhat oversimplified. It definitely might be the case that for some set $A \subseteq [n]$ of size at most $r-1$, the sets of $\calF$ that are disjoint from $A$ are quite rare, and as such, our random samples are not expected to include them. Following the terminology of~\cite{DinurF09}, we say that such a set $A$ captures $\calF$ (see Definition~\ref{def:rebellious}). To deal with this situation, we show that if a set $A$ of size at most $r-1$ captures $\calF$, then it admits two disjoint subsets $B,C \subseteq A$, such that the sub-families $\calF_1$ and $\calF_2$ of $\calF$, defined as the restrictions of $\calF$ to the sets whose intersection with $A$ is, respectively, $B$ and $C$, are far from being cross-intersecting, that is, one has to remove plenty of sets from $\calF_1$ and $\calF_2$ so that every set of $\calF_1$ will intersect every set of $\calF_2$. Since $B$ and $C$ are disjoint, this essentially allows us to ignore from now on the elements of $A$, and to analyze the probability that the random samples include two disjoint sets, one from $\calF_1$ and one from $\calF_2$. By a slight variant of the key observation described above, we wish our samples to include, for each set $A' \subseteq [n] \setminus A$ of size at most $r-1$, a set $F_{A'} \in \calF_2$ with $F_{A'} \cap A' = \emptyset$. If the samples include such sets, then it can be shown that a random set is expected to lie in $\calF_1$ and be disjoint from some previously chosen set of $\calF_2$ with non-negligible probability. However, if some set $A' \subseteq [n] \setminus A$ of size at most $r-1$ captures $\calF_2$, then a suitable set $F_{A'}$ is unlikely to be found. In this case, we repeat the above process and further refine $\calF_1$ and $\calF_2$, fixing the intersections of their sets with $A'$ to some disjoint subsets and keeping them far from cross-intersecting. It might be needed to repeat this process multiple times, but a crucial component of our argument shows that $r$ iterations suffice. Indeed, after this number of iterations, it turns out that the refined family $\calF_2$ is already too small to still satisfy the invariant that $\calF_1$ and $\calF_2$ are far from cross-intersecting.

To summarize, our analysis of the canonical tester shows that if a family $\calF \subseteq \binom{[n]}{k}$ is $\eps$-far from intersecting for a sufficiently large $\eps = \eps(n,k)$, then there exist a set $A \subseteq [n]$ and two disjoint subsets $B,C \subseteq A$, such that (a) the restrictions $\calF_1$ and $\calF_2$ of $\calF$ to the sets whose intersections with $A$ are $B$ and $C$, respectively, are far from cross-intersecting, and (b) no subset of $[n] \setminus A$ of size at most $r-1$ captures $\calF_2$ (see Lemma~\ref{lemma:ABCinduction}). This allows us to show that $O(\frac{\ln k}{\eps})$ sets chosen randomly from $\binom{[n]}{k}$ include with high probability a collection of sets of $\calF_2$, such that relatively few sets of $\calF_1$ intersect them all. Therefore, a few additional samples from $\binom{[n]}{k}$ are expected to include a set of $\calF_1$ that is disjoint from one of the previously picked sets of $\calF_2$, resulting in the desired witness for non-intersectingness (see Lemma~\ref{lemma:samples}).
For the precise and full argument, the reader is referred to Section~\ref{sec:GeneralR}.

\subsection{Related Work}

The Erd{\H{o}}s--Ko--Rado theorem~\cite{EKR61} implies that for all integers $n$ and $k$ with $n \geq 2k$, every family $\calF \subseteq \binom{[n]}{k}$ whose size exceeds $\binom{n-1}{k-1}$ is not intersecting. One may thus ask whether such a family $\calF$ must include a set that is disjoint from many of the sets of $\calF$.
This question was recently investigated by Frankl and Kupavskii~\cite{FranklK20} and by Chau, Ellis, Friedgut, and Lifshitz~\cite{CEFL23}, who provided a positive answer in a strong sense. Namely, it was shown in~\cite{CEFL23} that for any $\delta >0$, there exists some $\alpha >0$, such that for all integers $n$ and $k$ with $n \geq \alpha \cdot k$, every family $\calF \subseteq \binom{[n]}{k}$ of size $|\calF| = \binom{n-1}{k-1} +1$ includes a set that is disjoint from at least $(\frac{1}{2}-\delta) \cdot \binom{n-k-1}{k-1}$ of the sets of $\calF$.

We point out that the combinatorial question arising in our analysis of the canonical tester for the $\Inter_{n,k\eps}$ problem is similar in spirit to the question studied in~\cite{FranklK20,CEFL23}.
Indeed, the analysis requires us to show that if a family $\calF \subseteq \binom{[n]}{k}$ is $\eps$-far from intersecting, then a relatively small collection of random sets, chosen uniformly and independently from $\binom{[n]}{k}$, includes with high probability two disjoint sets that lie in $\calF$. For such a statement, one has to show not only the existence of a set in $\calF$ that is disjoint from many of the sets of $\calF$, but that many of the sets of $\calF$ satisfy this property. Yet, a crucial difference between the two concepts is that in ours, the family $\calF$ is not assumed to exceed the Erd{\H{o}}s--Ko--Rado threshold but rather to be $\eps$-far from intersecting.

\subsection{Outline}
The rest of the paper is organized as follows.
In Section~\ref{sec:prel}, we provide some definitions and tools that will be used throughout the paper.
In Section~\ref{sec:two-sided}, we present and analyze our two-sided error tester for the $\Inter_{n,k,\eps}$ problem and confirm Theorem~\ref{thm:1Intro}.
In Section~\ref{sec:one-sided}, we present and analyze the one-sided error canonical tester for the $\Inter_{n,k,\eps}$ problem and confirm Theorems~\ref{thm:2Intro} and~\ref{thm:3Intro}.
We further present there another one-sided error tester, appropriate for integers $n$ and $k$ with $n = \Theta(k)$.
Finally, in Section~\ref{sec:lower}, we complement our algorithmic results by proving the lower bound stated in Theorem~\ref{thm:LBIntro}.

\section{Preliminaries}\label{sec:prel}

Throughout the paper, we omit all floor and ceiling signs, whenever these are not crucial.

\subsection{Intersecting Families}

For integers $n$ and $k$, let $\binom{[n]}{k}$ denote the family of all $k$-subsets of $[n]=\{1,\ldots,n\}$.
A family $\calF \subseteq \binom{[n]}{k}$ is called intersecting if for every two sets $F_1,F_2 \in \calF$, it holds that $F_1 \cap F_2 \neq \emptyset$.
For a real $\eps \in [0,1]$, we say that $\calF$ is $\eps$-close to intersecting if it is possible to make $\calF$ intersecting by removing at most $\eps \cdot \binom{n}{k}$ of its sets.
Otherwise, we say that $\calF$ is $\eps$-far from intersecting.

For two families $\calF_1,\calF_2 \subseteq \binom{[n]}{k}$, the pair $(\calF_1,\calF_2)$ is called cross-intersecting if for every two sets $F_1 \in \calF_1$ and $F_2 \in \calF_2$, it holds that $F_1 \cap F_2 \neq \emptyset$. For a real $\eps \in [0,1]$, we say that the pair $(\calF_1,\calF_2)$ is $\eps$-close to cross-intersecting if it is possible to make $(\calF_1,\calF_2)$ cross-intersecting by removing at most $\eps \cdot \binom{n}{k}$ of the sets of $\calF_1$ and $\calF_2$.
Otherwise, we say that $(\calF_1,\calF_2)$ is $\eps$-far from cross-intersecting.
Note that if a family $\calF \subseteq \binom{[n]}{k}$ is $\eps$-far from intersecting, then the pair $(\calF,\calF)$ is $\eps$-far from cross-intersecting.

A function $f: \binom{[n]}{k} \rightarrow \{0,1\}$ is called intersecting, $\eps$-close to intersecting, and $\eps$-far from intersecting, if the family $f^{-1}(1)$ of the sets of $\binom{[n]}{k}$ that are mapped by $f$ to $1$ is, respectively, intersecting, $\eps$-close to intersecting, and $\eps$-far from intersecting.
For reals $\eps_1,\eps_2 \in [0,1)$ with $\eps_1 \leq \eps_2$, let $\Inter_{n,k,\eps_1,\eps_2}$ denote the tolerant property testing problem that given query access to a function $f: \binom{[n]}{k} \rightarrow \{0,1\}$, asks to distinguish between the case that $f$ is $\eps_1$-close to intersecting and the case that $f$ is $\eps_2$-far from intersecting.
We will be particularly interested in the case of $\eps_1=0$, hence for a real $\eps \in [0,1)$, we let $\Inter_{n,k,\eps}$ denote the $\Inter_{n,k,0,\eps}$ problem.

\subsection{Chernoff--Hoeffding Bound}

We will need the following version of the Chernoff--Hoeffding bound.
\begin{theorem}\label{thm:Chernoff}
For an integer $m$ and a real $p \in (0,1)$, let $X_1, \ldots, X_m$ be independent binary random variables satisfying $\Prob{}{X_i=1}=p$ and $\Prob{}{X_i=0}=1-p$ for all $i \in [m]$, and put $\overline{X} = \frac{1}{m} \cdot \sum_{i=1}^{m}{X_i}$. Then, for any $\mu \geq 0$,
\begin{enumerate}
  \item\label{itm1:Chernoff} if $p \leq \mu$, then $\Prob{}{~\overline{X} \geq 2 \mu~} \leq e^{- m \cdot \mu/3}$, and
  \item\label{itm2:Chernoff} if $p \geq \mu$, then $\Prob{}{~\overline{X} \leq \mu/2~} \leq e^{- m \cdot \mu/8}$.
\end{enumerate}
\end{theorem}
\noindent
Note that the assertion of Theorem~\ref{thm:Chernoff} for $p=\mu$ follows from standard statements of the Chernoff--Hoeffding bound (see, e.g,~\cite[Theorem~2.3]{McDiarmid98}).
The cases of $p < \mu$ and $p > \mu$ stem by monotonicity.

\section{Two-Sided Error Tester}\label{sec:two-sided}

In this section, we present and analyze a tolerant non-adaptive two-sided error tester for the $\Inter_{n,k,\eps_1,\eps_2}$ problem and prove the following result.
Its special case with $\eps_1 = 0$ yields Theorem~\ref{thm:1Intro}.

\begin{theorem}\label{thm:two-sided}
For every integer $r \geq 2$, there exist constants $c_1 = c_1(r)$ and $c_2 = c_2(r)$ for which the following holds.
For all sufficiently large integers $n$ and $k$ with $n \geq 2 k$ and for all reals $\eps_1,\eps_2 \in [0,1)$ with $\eps_2 \geq 4 \cdot \eps_1 + c_1 \cdot (\frac{k}{n})^r$, there exists a tolerant non-adaptive two-sided error tester for $\Inter_{n,k,\eps_1,\eps_2}$ with at most $c_2 \cdot \frac{\ln n}{\eps_2}$ queries, running time polynomial in $n$, and success probability at least $2/3$.
\end{theorem}

A crucial ingredient in the proof of Theorem~\ref{thm:two-sided} is the following theorem due to Dinur and Friedgut~\cite{DinurF09}.
Here, a family $\calJ$ of subsets of $[n]$ is called a $j$-junta over $[n]$ if there exists a set $J \subseteq [n]$ of size $|J|=j$ such that the membership of a set $F$ in $\calJ$ depends only on $F \cap J$.
\begin{theorem}[\cite{DinurF09}]\label{thm:DF}
For every integer $r \geq 2$, there exist constants $j = j(r)$ and $a = a(r)$ for which the following holds.
For all integers $n$ and $k$ with $n > 2 k$ and $k > j$, if a family $\calF \subseteq \binom{[n]}{k}$ is intersecting, then there exists an intersecting $j$-junta $\calJ$ over $[n]$ such that $|\calF \setminus \calJ | \leq a \cdot \binom{n-r}{k-r}$.
\end{theorem}

We are ready to prove Theorem~\ref{thm:two-sided}.

\begin{proof}[ of Theorem~\ref{thm:two-sided}]
Fix an integer $r \geq 2$, and let $j=j(r)$ and $a=a(r)$ be the constants given in Theorem~\ref{thm:DF}.
Let $n$ and $k$ be sufficiently large integers. Since the theorem trivially holds for $n=2k$ (with an appropriate $c_1$), it may be assumed that $n > 2k$ and $k > j$.
For an integer $m$ and for reals $\eps_1,\eps_2 \in [0,1)$ with $\eps_1 \leq \eps_2$, consider the following tester for the $\Inter_{n,k,\eps_1,\eps_2}$ problem.
\begin{tcolorbox}
Input: Query access to a function $f: \binom{[n]}{k} \rightarrow \{0,1\}$.
\begin{enumerate}
  \item Pick $m$ sets $G_1, \ldots, G_m$ uniformly and independently at random from $\binom{[n]}{k}$.
  \item Query $f$ for the value of $f(G_i)$ for each $i \in [m]$.
  \item For each intersecting $j$-junta $\calJ$ over $[n]$, let $\alpha_\calJ$ denote the fraction of the chosen sets $G_i$ that are mapped by $f$ to $1$ and do not lie in $\calJ$, that is, \[ \alpha_\calJ = \tfrac{1}{m} \cdot |\{i \in [m] \mid f(G_i)=1~\mbox{and}~G_i \notin  \calJ\}|.\]
  \item If there exists an intersecting $j$-junta $\calJ$ over $[n]$ with $\alpha_\calJ \leq \frac{\eps_2}{2}$, then accept, and otherwise reject.
\end{enumerate}
\end{tcolorbox}

The above tester is clearly non-adaptive.
For the analysis of its success probability, suppose that
\begin{eqnarray}\label{eq:eps,m_two-sided}
\eps_2 \geq 4 \cdot \bigg ( \eps_1 + a \cdot \Big (\frac{k}{n} \Big )^r \bigg )~~~~\mbox{and}~~~~m \geq \frac{12 \cdot (j \cdot \ln n +2^{j}+2)}{\eps_2}.
\end{eqnarray}
Let $f: \binom{[n]}{k} \rightarrow \{0,1\}$ be an input function, and consider the family $\calF = f^{-1}(1)$.
Our goal is to prove that if $\calF$ is $\eps_1$-close to intersecting then the tester accepts $f$ with probability at least $2/3$, and that if $\calF$ is $\eps_2$-far from intersecting then the tester rejects $f$ with probability at least $2/3$.

For each intersecting $j$-junta $\calJ$ over $[n]$, let $p_\calJ$ denote the fraction of the sets of $\binom{[n]}{k}$ that are mapped by $f$ to $1$ and do not lie in $\calJ$, that is, \[ p_\calJ = \frac{1}{\binom{n}{k}} \cdot \bigg | \bigg \{G \in \binom{[n]}{k} ~\bigg | ~f(G)=1~\mbox{and}~G \notin \calJ \bigg \} \bigg |.\]
We shall prove that, with high probability, the values of $p_\calJ$ are well approximated by the quantities $\alpha_\calJ$ that our tester computes.
Let $\calE$ denote the event that for every intersecting $j$-junta $\calJ$ over $[n]$, it holds that
\begin{enumerate}
  \item if $p_\calJ \leq \frac{\eps_2}{4}$, then $\alpha_\calJ < \frac{\eps_2}{2}$, and
  \item if $p_\calJ \geq \eps_2$, then $\alpha_\calJ > \frac{\eps_2}{2}$.
\end{enumerate}
For a fixed intersecting $j$-junta $\calJ$ over $[n]$, apply Item~\ref{itm1:Chernoff} of Theorem~\ref{thm:Chernoff} with $\mu = \frac{\eps_2}{4}$ to obtain that if $p_\calJ \leq \frac{\eps_2}{4}$, then $\alpha_\calJ \geq \frac{\eps_2}{2}$ with probability at most $e^{-m \cdot \eps_2 /12}$. Further, apply Item~\ref{itm2:Chernoff} of Theorem~\ref{thm:Chernoff} with $\mu = \eps_2$ to obtain that if $p_\calJ \geq \eps_2$, then $\alpha_\calJ \leq \frac{\eps_2}{2}$ with probability at most $e^{-m \cdot \eps_2 /8} \leq e^{-m \cdot \eps_2 /12}$.
The number of intersecting $j$-juntas over $[n]$ is clearly bounded by $\binom{n}{j} \cdot 2^{2^j}$.
Therefore, by the union bound, the probability that there exists an intersecting $j$-junta $\calJ$ over $[n]$ for which either $p_\calJ \leq \frac{\eps_2}{4}$ and $\alpha_\calJ \geq \frac{\eps_2}{2}$, or $p_\calJ \geq \eps_2$ and $\alpha_\calJ \leq \frac{\eps_2}{2}$ does not exceed
\[\binom{n}{j} \cdot 2^{2^j} \cdot e^{-m \cdot \eps_2 / 12} \leq n^j \cdot  2^{2^j} \cdot e^{-m \cdot \eps_2 / 12} \leq e^{-2} < \tfrac{1}{3},\]
where the second inequality holds by our assumption on $m$ given in~\eqref{eq:eps,m_two-sided}.
Therefore, the event $\calE$ occurs with probability at least $2/3$.

It suffices to show now that if the event $\calE$ occurs, then the answer of our tester is correct.
Suppose first that $\calF$ is $\eps_1$-close to intersecting.
By definition, there exists an intersecting family $\calF' \subseteq \binom{[n]}{k}$ such that $|\calF \setminus \calF'| \leq \eps_1 \cdot \binom{n}{k}$.
By Theorem~\ref{thm:DF}, using $n > 2k$ and $k>j$, there exists an intersecting $j$-junta $\calJ$ over $[n]$ such that $|\calF' \setminus \calJ | \leq a \cdot \binom{n-r}{k-r} \leq a \cdot (\frac{k}{n} )^r \cdot \binom{n}{k}$. It thus follows that this $\calJ$ satisfies
\[ |\calF \setminus \calJ | \leq \eps_1 \cdot \binom{n}{k} + |\calF' \setminus \calJ| \leq \eps_1 \cdot \binom{n}{k} + a \cdot \Big (\frac{k}{n} \Big )^r \cdot \binom{n}{k} \leq  \frac{\eps_2}{4} \cdot \binom{n}{k},\]
where the last inequality relies on our assumption on $\eps_1$ and $\eps_2$ in~\eqref{eq:eps,m_two-sided}.
This implies that $p_\calJ \leq \frac{\eps_2}{4}$, and since the event $\calE$ occurs, it follows that $\alpha_\calJ < \frac{\eps_2}{2}$, hence our tester accepts $f$.
Next, suppose that $\calF$ is $\eps_2$-far from intersecting. This implies, for each intersecting $j$-junta $\calJ$ over $[n]$, that $p_\calJ > \eps_2$, as otherwise, one could remove at most $\eps_2 \cdot \binom{n}{k}$ of the sets of $\calF$ to obtain a sub-family of the intersecting family $\calJ$. Since the event $\calE$ occurs, each such $\calJ$ satisfies $\alpha_\calJ > \frac{\eps_2}{2}$, hence our tester rejects $f$.

Finally, let $m$ be the smallest integer satisfying the condition in~\eqref{eq:eps,m_two-sided}.
Observe that the running time of our tester is polynomial in $m$ and in the number of intersecting $j$-juntas over $[n]$, where $j$ depends only on $r$. It follows that the running time is polynomial in $n$ and $1/\eps_2$, so by our assumption on $\eps_2$ in~\eqref{eq:eps,m_two-sided}, it is polynomial in $n$. This completes the proof.
\end{proof}

\begin{remark}
The assumption on $\eps_1$ and $\eps_2$ in Theorem~\ref{thm:two-sided} can be weakened, by an almost identical proof, to $\eps_2 \geq (1+\delta) \cdot \eps_1 + c_1 \cdot (\frac{k}{n})^r$ with an arbitrary $\delta >0$. For simplicity of presentation, we omit the details.
\end{remark}

We conclude this section with the observation that for all integers $n$ and $k$ with $n \geq 2k$ and for all reals $\eps_1,\eps_2 \in [0,1)$ with $\eps_2 \geq 4 \cdot (\eps_1+\frac{k}{n})$, there exists an efficient tolerant non-adaptive two-sided error tester for $\Inter_{n,k,\eps_1, \eps_2}$ with $O(\frac{1}{\eps_2})$ queries and high success probability. To see this, consider the tester that given access to a function $f: \binom{[n]}{k} \rightarrow \{0,1\}$ picks $O(\frac{1}{\eps_2})$ sets uniformly and independently at random from $\binom{[n]}{k}$, queries $f$ on them, and computes the fraction $\alpha$ of the chosen sets that are mapped by $f$ to $1$. If $\alpha \leq \frac{\eps_2}{2}$, then the tester accepts, and otherwise it rejects.

For correctness, let $p$ denote the fraction of the sets of $\binom{[n]}{k}$ that are mapped by $f$ to $1$.
If $f$ is $\eps_1$-close to intersecting, then $f$ can be made intersecting by flipping at most $\eps_1 \cdot \binom{n}{k}$ of its values. By the Erd{\H{o}}s--Ko--Rado theorem~\cite{EKR61}, the fraction of the sets of $\binom{[n]}{k}$ in an intersecting family does not exceed $\frac{k}{n}$, hence $p \leq \eps_1+\frac{k}{n} \leq \frac{\eps_2}{4}$. In this case, Theorem~\ref{thm:Chernoff} yields that with high probability, it holds that $\alpha \leq \frac{\eps_2}{2}$ and the tester accepts. On the other hand, if $f$ is $\eps_2$-far from intersecting, then it holds that $p > \eps_2$, so using again Theorem~\ref{thm:Chernoff}, it follows that with high probability, we have $\alpha > \frac{\eps_2}{2}$ and the tester rejects.

\section{One-Sided Error Tester}\label{sec:one-sided}

In this section, we present non-adaptive one-sided error testers for the $\Inter_{n,k,\eps}$ problem.
We start with the canonical tester, which is used to prove Theorems~\ref{thm:2Intro} and~\ref{thm:3Intro}.
We then offer another non-adaptive one-sided error tester, appropriate for integers $n$ and $k$ with $n = \Theta(k)$.

\subsection{Canonical Tester}
Consider the following tester for the $\Inter_{n,k,\eps}$ problem.
\begin{tcolorbox}
$\Sampler~(n,k,m)$ \\
Input: Query access to a function $f: \binom{[n]}{k} \rightarrow \{0,1\}$.
\begin{enumerate}
  \item Pick $m$ sets $G_1, \ldots, G_m$ uniformly and independently at random from $\binom{[n]}{k}$.
  \item Query $f$ for the value of $f(G_i)$ for each $i \in [m]$.
  \item If for every two indices $i_1,i_2 \in [m]$ with $G_{i_1} \cap G_{i_2} = \emptyset$, it holds that either $f(G_{i_1})=0$ or $f(G_{i_2})=0$, then accept, and otherwise reject.
\end{enumerate}
\end{tcolorbox}

The description of $\Sampler$ immediately implies that it is non-adaptive and that it accepts every intersecting function with probability $1$ and is thus one-sided error.
In what follows, we analyze the number of queries $m$ needed to reject with high probability any function $\eps$-far from intersecting for $\eps \geq \Omega( (\frac{k^2}{n})^r )$, where $r$ is a fixed integer.
We start with the case of $r=2$, for which we offer a simple analysis that yields a bound of $O(\frac{1}{\eps})$ on the query complexity.
We then consider the case of a general $r$, for which we obtain a bound of $O(\frac{\ln k}{\eps})$ on the query complexity.

\subsubsection{The case \texorpdfstring{$r=2$}{r=2}}

We prove the following result, which confirms Theorem~\ref{thm:3Intro}.

\begin{theorem}\label{thm:r=2}
For all integers $n$ and $k$ with $n \geq 2k$ and for any real $\eps \in [0,1)$ with $\eps \geq 2 \cdot (\frac{k^2}{n})^2$, there exists a non-adaptive one-sided error tester for $\Inter_{n,k,\eps}$ with $O(\frac{1}{\eps})$ queries, running time polynomial in $n$, and success probability at least $2/3$.
\end{theorem}

The proof of Theorem~\ref{thm:r=2} requires two simple lemmas.
The first one, given below, may be viewed as an analogue of a lemma of~\cite{CDLNS24} for uniform families.
Here, for a collection $\calM$ of pairs of sets, we say that the pairs of $\calM$ are pairwise disjoint if no set lies in more than one pair of $\calM$.

\begin{lemma}\label{lemma:M}
For integers $n$ and $k$ with $n \geq 2k$ and for a real $\eps \in [0,1)$,
if a family $\calF \subseteq \binom{[n]}{k}$ is $\eps$-far from intersecting, then there exists a collection of more than $\frac{\eps}{2} \cdot \binom{n}{k}$ pairwise disjoint pairs $(C,D)$ of sets of $\calF$ satisfying $C \cap D = \emptyset$.
\end{lemma}

\begin{proof}
Let $\calF \subseteq \binom{[n]}{k}$ be a family $\eps$-far from intersecting, and consider a maximal collection $\calM$ (with respect to inclusion) of pairwise disjoint pairs $(C,D)$ of sets of $\calF$ satisfying $C \cap D = \emptyset$.
The maximality of $\calM$ implies that $\calF$ can be made intersecting by removing all the $2 \cdot |\calM|$ sets that lie in the pairs of $\calM$.
Since $\calF$ is $\eps$-far from intersecting, it follows that $2 \cdot |\calM| > \eps \cdot \binom{n}{k}$. This completes the proof.
\end{proof}

Using Lemma~\ref{lemma:M}, we obtain the following.

\begin{lemma}\label{lemma:useful}
For integers $n$ and $k$ with $n \geq 2k$ and for a real $\eps \in [0,1)$, let $\calF \subseteq \binom{[n]}{k}$ be a family $\eps$-far from intersecting with $|\calF| > k^2 \cdot \binom{n-2}{k-2}$.
Then, more than $\frac{\eps}{2} \cdot \binom{n}{k}$ of the sets of $\calF$ are disjoint from at least $\frac{1}{2} \cdot ( |\calF| - k^2 \cdot \binom{n-2}{k-2})$ of the sets of $\calF$.
\end{lemma}

\begin{proof}
Let $\calF \subseteq \binom{[n]}{k}$ be a family $\eps$-far from intersecting.
By Lemma~\ref{lemma:M}, there exists a collection $\calM$ of more than $\frac{\eps}{2} \cdot \binom{n}{k}$ pairwise disjoint pairs $(C,D)$ of sets of $\calF$ satisfying $C \cap D = \emptyset$. Fix a pair $(C,D) \in \calM$. The disjointness of $C$ and $D$ implies that every set in $\binom{[n]}{k}$ that intersects both $C$ and $D$ includes some element of $C$ and a different element of $D$. It follows that the number of sets in $\binom{[n]}{k}$ that intersect both $C$ and $D$ does not exceed $k^2 \cdot \binom{n-2}{k-2}$. Therefore, either $C$ or $D$ is disjoint from at least $\frac{1}{2} \cdot ( |\calF| - k^2 \cdot \binom{n-2}{k-2})$ of the sets of $\calF$.
Since the pairs of $\calM$ are pairwise disjoint, it follows that more than $\frac{\eps}{2} \cdot \binom{n}{k}$ sets of $\calF$ are disjoint from at least $\frac{1}{2} \cdot ( |\calF| - k^2 \cdot \binom{n-2}{k-2})$ of the sets of $\calF$, as required.
\end{proof}

We are ready to prove Theorem~\ref{thm:r=2}.

\begin{proof}[ of Theorem~\ref{thm:r=2}]
For integers $n,k,m$ and for a real $\eps \in [0,1)$, consider the $\textsc{Canonical}$ $\textsc{Tester}~(n,k,m)$ for the $\Inter_{n,k,\eps}$ problem.
As previously mentioned, this tester is non-adaptive and one-sided error.
To establish its correctness, it suffices to prove that if a function $f: \binom{[n]}{k} \rightarrow \{0,1\}$ is $\eps$-far from intersecting, then the $m$ sets picked by the tester include with probability at least $2/3$ two disjoint sets that lie in $f^{-1}(1)$. Indeed, this implies that the tester rejects such an $f$ with the desired probability.

For a real $\eps \geq 2 \cdot (\frac{k^2}{n})^2$, let $f: \binom{[n]}{k} \rightarrow \{0,1\}$ be a function $\eps$-far from intersecting. Consider the family $\calF = f^{-1}(1)$, and note that $\calF$ is $\eps$-far from intersecting. In particular, it holds that $|\calF| > \eps \cdot \binom{n}{k}$, because $\calF$ can be made intersecting by removing all of its sets. This implies that
\begin{eqnarray}\label{eq:useful}
|\calF| - k^2 \cdot \binom{n-2}{k-2} > \eps \cdot \binom{n}{k} - k^2 \cdot \bigg (\frac{k}{n} \bigg )^2 \cdot \binom{n}{k} \geq \frac{\eps}{2} \cdot \binom{n}{k},
\end{eqnarray}
where the second inequality relies on our assumption on $\eps$.
We say that a set of $\binom{[n]}{k}$ is useful if it lies in $\calF$ and is disjoint from at least $\frac{\eps}{4} \cdot \binom{n}{k}$ of the sets of $\calF$.
By Lemma~\ref{lemma:useful}, combined with~\eqref{eq:useful}, the number of useful sets is greater than $\frac{\eps}{2} \cdot \binom{n}{k}$. Hence, a random set picked uniformly from $\binom{[n]}{k}$ is useful with probability at least $\eps/2$.

Now, consider the $m$ sets chosen by $\Sampler$.
The probability that no useful set is picked throughout the first $4/\eps$ choices does not exceed $(1-\eps/2)^{4/\eps} \leq e^{-2} < 1/6$.
Therefore, with probability at least $5/6$, these choices include a useful set $A$.
Since $A$ is useful, it is disjoint from at least $\frac{\eps}{4} \cdot \binom{n}{k}$ of the sets of $\calF$.
Therefore, once such a set $A$ is chosen, the probability that a random set picked uniformly from $\binom{[n]}{k}$ lies in $\calF$ and is disjoint from $A$ is at least $\eps/4$. It follows that the probability that no such set is picked throughout the next $8/\eps$ choices does not exceed $(1-\eps/4)^{8/\eps} \leq e^{-2} < 1/6$.
By the union bound, for $m = 12/\eps$, the $m$ sets picked by our tester include with probability at least $2/3$ two disjoint sets that lie in $\calF$.
Observe that for this choice of $m$, the running time of our tester is polynomial in $n$ and $1/\eps$.
By our assumption on $\eps$, the running time is polynomial in $n$, as desired.
\end{proof}

\subsubsection{General \texorpdfstring{$r$}{r}}\label{sec:GeneralR}

We proceed by analyzing the more nuanced case in which $\eps \geq \Omega( (k^2/n)^r)$ for a general integer $r$.
We prove the following result, which confirms Theorem~\ref{thm:2Intro}.

\begin{theorem}\label{thm:general_r}
For every integer $r \geq 1$, there exist constants $c_1 = c_1(r)$ and $c_2 = c_2(r)$ for which the following holds.
For all integers $n$ and $k$ with $n \geq 2k$ and for any real $\eps \in [0,1)$ with $\eps \geq c_1 \cdot (\frac{k^2}{n})^r$, there exists a non-adaptive one-sided error tester for $\Inter_{n,k,\eps}$ with at most $c_2 \cdot \frac{\ln k}{\eps}$ queries, running time polynomial in $n$, and success probability at least $2/3$.
\end{theorem}

The proof of Theorem~\ref{thm:general_r} requires a few lemmas, which involve the following definition.

\begin{definition}\label{def:rebellious}
For integers $n$ and $k$ with $n \geq 2k$, a family $\calF \subseteq \binom{[n]}{k}$, a set $A \subseteq [n]$, and a set $B \subseteq A$, we let $\calF(A_{\downarrow B})$ denote the family of sets of $\calF$ whose intersection with $A$ is $B$, that is,
\[\calF(A_{\downarrow B}) = \{ F \in \calF \mid F \cap A = B \}.\]
We say that the set $A$ $\eps$-captures $\calF$ if the number of sets of $\calF$ that are disjoint from $A$ is smaller than $\eps \cdot \binom{n}{k}$, equivalently, $|\calF(A_{\downarrow \emptyset})| < \eps \cdot \binom{n}{k}$.
\end{definition}

The following lemma shows that if two families are far from cross-intersecting, then for every small set $A$, it is possible to restrict the families to disjoint intersections with $A$, so that the obtained restrictions are still far from cross-intersecting.

\begin{lemma}\label{lemma:B12}
For an integer $r \geq 0$, integers $n$ and $k$ with $n \geq 2k$, and a real $\eps \in [0,1)$, let $\calF_1,\calF_2 \subseteq \binom{[n]}{k}$ be two families such that the pair $(\calF_1,\calF_2)$ is $\eps$-far from cross-intersecting, and let $A \subseteq [n]$ be a set of size $|A| \leq r$.
Then, there exist two sets $B,C \subseteq A$ with $B \cap C = \emptyset$ for which the pair $(\calF_1(A_{\downarrow B}), \calF_2(A_{\downarrow C}))$ is $\frac{\eps}{3^r}$-far from cross-intersecting.
Moreover, if $A$ $\frac{\eps}{3^r}$-captures $\calF_2$, then the guaranteed set $C$ is not empty.
\end{lemma}

\begin{proof}
Let $\calF_1, \calF_2 \subseteq \binom{[n]}{k}$ and $A \subseteq [n]$ be as in the statement of the lemma.
It suffices to prove the lemma for the case where $|A|=r$, because if $|A| = r' < r$, then one can apply the lemma with $r$ being $r'$ to obtain a stronger statement.

Assuming that $|A|=r$, the number of pairs $(B,C)$ of disjoint subsets of $A$ is $3^r$.
Suppose for contradiction that for every such pair $(B,C)$, the pair $(\calF_1(A_{\downarrow B}), \calF_2(A_{\downarrow C}))$ is $\frac{\eps}{3^r}$-close to cross-intersecting, that is, it can be made cross-intersecting by removing at most $\frac{\eps}{3^r} \cdot \binom{n}{k}$ sets from $\calF_1(A_{\downarrow B})$ and $\calF_2(A_{\downarrow C})$. This implies that by removing at most $\eps \cdot \binom{n}{k}$ sets from $\calF_1$ and $\calF_2$, it is possible to obtain two sub-families $\calF'_1 \subseteq \calF_1$ and $\calF'_2 \subseteq \calF_2$, such that for every two disjoint (possibly empty) sets $B,C \subseteq A$, the pair $(\calF'_1(A_{\downarrow B}), \calF'_2(A_{\downarrow C}))$ is cross-intersecting. Since for every non-disjoint sets $B,C \subseteq A$, the pair $(\calF'_1(A_{\downarrow B}), \calF'_2(A_{\downarrow C}))$ is also cross-intersecting (including the case $B=C \neq \emptyset$), it follows that the pair $(\calF'_1,\calF'_2)$ is cross-intersecting. This contradicts the assumption that $(\calF_1,\calF_2)$ is $\eps$-far from cross-intersecting.

Now, let $(B,C)$ be a pair of disjoint subsets of $A$ for which $(\calF_1(A_{\downarrow B}), \calF_2(A_{\downarrow C}))$ is $\frac{\eps}{3^r}$-far from cross-intersecting.
The distance from cross-intersecting yields that $|\calF_2(A_{\downarrow C})| > \frac{\eps}{3^r} \cdot \binom{n}{k}$, because the pair can be made cross-intersecting by removing all the sets of $\calF_2(A_{\downarrow C})$. Therefore, when $A$ $\frac{\eps}{3^r}$-captures $\calF_2$, the set $C$ is not empty, and we are done.
\end{proof}

\begin{remark}
In the proof of Lemma~\ref{lemma:B12}, it suffices to consider the unordered pairs of disjoint subsets of $A$, whose number is $\tfrac{1}{2} \cdot (3^r+1)$. We use the weaker bound $3^r$ for simplicity of presentation.
\end{remark}

As a consequence of Lemma~\ref{lemma:B12}, we obtain the following.

\begin{lemma}\label{lemma:ABCinduction}
For integers $r,t \geq 0$, integers $n$ and $k$ with $n \geq 2k$, and a real $\eps \in [0,1)$, let $\calF_1,\calF_2 \subseteq \binom{[n]}{k}$ be families such that the pair $(\calF_1,\calF_2)$ is $\eps$-far from cross-intersecting.
Then, there exist sets $A \subseteq [n]$ and $B, C \subseteq A$ with $B \cap C = \emptyset$, such that
\begin{enumerate}
  \item\label{itm:1} the pair $(\calF_1(A_{\downarrow B}),\calF_2(A_{\downarrow C}))$ is $\frac{\eps}{3^{r \cdot t}}$-far from cross-intersecting, and
  \item\label{itm:2} either there is no subset of $[n] \setminus A$ of size at most $r$ that $\frac{\eps}{3^{r \cdot t}}$-captures $\calF_2(A_{\downarrow C})$, or $|C| \geq t$.
\end{enumerate}
Moreover, if $\eps \geq 3^{r\cdot t} \cdot (\frac{k}{n})^t$, then the guaranteed set $C$ satisfies $|C| < t$.\footnote{Note that by combining the second item of the lemma with the `moreover' part, it follows that if $\eps \geq 3^{r\cdot t} \cdot (\frac{k}{n})^t$, then there is no subset of $[n] \setminus A$ of size at most $r$ that $\frac{\eps}{3^{r \cdot t}}$-captures $\calF_2(A_{\downarrow C})$.}
\end{lemma}

\begin{proof}
Fix an integer $r \geq 0$.
We start by proving the first part of the lemma, i.e., the existence of sets $A \subseteq [n]$ and $B, C \subseteq A$ with $B \cap C = \emptyset$ satisfying Items~\ref{itm:1} and~\ref{itm:2}. To do so, we apply induction on $t$.
The result clearly holds for $t=0$, as follows from the choice $A = B = C = \emptyset$ (for which we have $\calF_1 = \calF_1(A_{\downarrow B})$, $\calF_2 = \calF_2(A_{\downarrow C})$, and $|C| \geq 0$).

Now, take $t \geq 1$, and assume that the result holds for $t-1$.
To prove it for $t$, let $\calF_1,\calF_2 \subseteq \binom{[n]}{k}$ be families such that the pair $(\calF_1,\calF_2)$ is $\eps$-far from cross-intersecting.
By the induction hypothesis, there exist sets $A \subseteq [n]$ and $B, C \subseteq A$ with $B \cap C = \emptyset$, such that the families $\calH_1 = \calF_1(A_{\downarrow B})$ and $ \calH_2 = \calF_2(A_{\downarrow C})$ satisfy that
\begin{enumerate}
  \item\label{itm:11} the pair $(\calH_1,\calH_2)$ is $\frac{\eps}{3^{r \cdot (t-1)}}$-far from cross-intersecting, and
  \item\label{itm:22} either there is no subset of $[n] \setminus A$ of size at most $r$ that $\frac{\eps}{3^{r \cdot (t-1)}}$-captures $\calH_2$, or $|C| \geq t-1$.
\end{enumerate}
By $\frac{\eps}{3^{r \cdot t}} \leq \frac{\eps}{3^{r \cdot (t-1)}}$, it follows from Item~\ref{itm:11} that the pair $(\calH_1,\calH_2)$ is $\frac{\eps}{3^{r \cdot t}}$-far from cross-intersecting. Therefore, if there is no subset of $[n] \setminus A$ of size at most $r$ that $\frac{\eps}{3^{r \cdot t}}$-captures $\calH_2$, then the result holds for $t$ with the same sets $A,B,C$, and we are done. Otherwise, there exists a set $A' \subseteq [n] \setminus A$ of size at most $r$ that $\frac{\eps}{3^{r \cdot t}}$-captures $\calH_2$. In particular, there exists a subset of $[n] \setminus A$ of size at most $r$ that $\frac{\eps}{3^{r \cdot (t-1)}}$-captures $\calH_2$, hence by Item~\ref{itm:22} above, it follows that $|C| \geq t-1$.

We apply Lemma~\ref{lemma:B12} with the pair $(\calH_1, \calH_2)$, which is $\frac{\eps}{3^{r \cdot (t-1)}}$-far from cross-intersecting, and with the set $A'$.
Letting $B',C' \subseteq A'$ be the sets with $B' \cap C' = \emptyset$ guaranteed by the lemma, it follows that the pair $(\calH_1(A'_{\downarrow B'}),\calH_2(A'_{\downarrow C'}))$ is $\tilde{\eps}$-far from cross-intersecting for
$\tilde{\eps} = \frac{1}{3^r} \cdot \frac{1}{3^{r \cdot (t-1)}} = \frac{1}{3^{r \cdot t}}$.
Moreover, since $A'$ $\tilde{\eps}$-captures $\calH_2$, it follows from the lemma that $C'$ is not empty.

To complete the argument, define $A'' = A \cup A'$, $B'' = B \cup B'$, and $C'' = C \cup C'$.
Recalling that $A \cap A' = \emptyset$, $B \cap C = \emptyset$, and $B' \cap C' = \emptyset$, we obtain that $B'' \cap C'' = \emptyset$ and that
\[|C''| = |C|+|C'| \geq (t-1)+1=t.\]
We further have $\calH_1(A'_{\downarrow B'}) = \calF_1(A''_{\downarrow B''})$ and $\calH_2(A'_{\downarrow C'}) = \calF_2(A''_{\downarrow C''})$.
Hence, the sets $A'',B'',C''$ satisfy the required properties with respect to $\calF_1,\calF_2$ and $t$. This completes the proof of the first part of the lemma.

We finally show that if $\eps \geq 3^{r\cdot t} \cdot (\frac{k}{n})^t$, then the guaranteed set $C$ satisfies $|C| < t$.
To establish the contrapositive statement, let $A,B,C$ be the sets guaranteed by the lemma for two families $\calF_1,\calF_2 \subseteq \binom{[n]}{k}$, and suppose that $|C| \geq t$.
Since the pair $(\calF_1(A_{\downarrow B}),\calF_2(A_{\downarrow C}))$ is $\frac{\eps}{3^{r \cdot t}}$-far from cross-intersecting, it follows that $| \calF_2(A_{\downarrow C}) | > \frac{\eps}{3^{r \cdot t}} \cdot \binom{n}{k}$.
On the other hand, each set of $\calF_2(A_{\downarrow C})$ contains the set $C$, hence $|\calF_2(A_{\downarrow C})| \leq \binom{n-t}{k-t} \leq (\frac{k}{n})^t \cdot \binom{n}{k}$.
By combining the two inequalities, we obtain that $\eps < 3^{r\cdot t} \cdot (\frac{k}{n})^t$. This completes the proof.
\end{proof}

The following lemma constitutes a key ingredient in our analysis of the canonical tester.

\begin{lemma}\label{lemma:samples}
For every integer $r \geq 1$, there exists a constant $c = c(r)$, such that for all integers $n$ and $k$ with $n \geq 2k$ and for any real $\eps \in [0,1)$ with $\eps \geq 2 \cdot (\frac{k^2}{n})^r$, the following holds.
Let $\calF \subseteq \binom{[n]}{k}$ be a family, and suppose that there exist sets $A \subseteq [n]$ and $B, C \subseteq A$ with $B \cap C = \emptyset$, such that the pair $(\calF(A_{\downarrow B}),\calF(A_{\downarrow C}))$ is $\eps$-far from cross-intersecting, and there is no subset of $[n] \setminus A$ of size at most $r-1$ that $\eps$-captures $\calF(A_{\downarrow C})$.
Then, if at least $c \cdot \frac{\ln k}{\eps}$ sets are chosen uniformly and independently at random from $\binom{[n]}{k}$, then the probability that they include two disjoint sets that lie in $\calF$ is at least $2/3$.
\end{lemma}

\begin{proof}
Consider a family $\calF$ and sets $A,B,C$ as in the statement of the lemma.
Put $\calF_1 = \calF(A_{\downarrow B})$ and $\calF_2 = \calF(A_{\downarrow C})$.
By assumption, the pair $(\calF_1,\calF_2)$ is $\eps$-far from cross-intersecting, and there is no subset of $[n] \setminus A$ of size at most $r-1$ that $\eps$-captures $\calF_2$.

Suppose that we pick sets uniformly and independently at random from $\binom{[n]}{k}$.
Since the empty set does not $\eps$-capture $\calF_2$, the probability that a random set from $\binom{[n]}{k}$ lies in $\calF_2$ is at least $\eps$.
Hence, the probability that no set of $\calF_2$ is chosen throughout the first $m_0 = \ln(6r)/\eps$ choices is at most
\[(1-\eps)^{m_0} \leq e^{-m_0 \cdot \eps} \leq \tfrac{1}{6r}.\]
Once a set $F \in \calF_2$ is chosen, for each $j_1 \in F \setminus A$, we are interested in a set $F_{j_1}$ that lies in $\calF_2$ and satisfies $F_{j_1} \cap \{j_1\} = \emptyset$. For $r \geq 2$, our assumptions imply that the set $\{j_1\}$ does not $\eps$-capture $\calF_2$, hence the probability that a random set from $\binom{[n]}{k}$ satisfies these conditions is at least $\eps$. Thus, by the union bound, the probability that for some $j_1 \in F \setminus A$, no set that lies in $\calF_2$ and is disjoint from $\{j_1\}$ is chosen throughout the next $m_1 = \ln(6rk)/\eps$ choices is at most
\[ k \cdot (1-\eps)^{m_1} \leq k \cdot e^{-m_1 \cdot \eps} \leq \tfrac{1}{6r}.\]
Once we are given sets $F_{j_1}$ as above for all $j_1 \in F \setminus A$, for every pair $(j_1,j_2) \in (F \setminus A) \times (F_{j_1} \setminus A)$, we are interested in a set $F_{j_1,j_2}$ that lies in $\calF_2$ and satisfies $F_{j_1,j_2} \cap \{j_1,j_2\} = \emptyset$. As before, for $r \geq 3$, the set $\{j_1,j_2\}$ does not $\eps$-capture $\calF_2$, hence the probability that a random set from $\binom{[n]}{k}$ satisfies these conditions is at least $\eps$. Thus, by the union bound, the probability that for some pair $(j_1,j_2) \in (F \setminus A) \times (F_{j_1} \setminus A)$, no set that lies in $\calF_2$ and is disjoint from $\{j_1,j_2\}$ is chosen throughout the next $m_2 = \ln(6rk^2)/\eps$ choices is at most
\[ k^2 \cdot (1-\eps)^{m_2} \leq k^2 \cdot e^{-m_2 \cdot \eps} \leq \tfrac{1}{6r}.\]
Proceeding in this way, in the $i$th iteration with $i \leq r-1$, for every $i$-tuple
\[(j_1,\ldots, j_i) \in (F \setminus A) \times (F_{j_1} \setminus A) \times (F_{j_1,j_2} \setminus A) \times \cdots \times (F_{j_1, \ldots, j_{i-1}} \setminus A),\]
we are interested in a set $F_{j_1, \ldots, j_i}$ that lies in $\calF_2$ and satisfies $F_{j_1, \ldots, j_i} \cap \{j_1, \ldots, j_{i}\} = \emptyset$.
Since the set $\{j_1, \ldots, j_{i}\}$ does not $\eps$-capture $\calF_2$, the probability that a random set from $\binom{[n]}{k}$ satisfies these conditions is at least $\eps$. Hence, by the union bound, the probability that for some $i$-tuple $(j_1,\ldots, j_i)$ as above, no set that lies in $\calF_2$ and is disjoint from $\{j_1, \ldots, j_{i}\}$ is chosen throughout the next $m_i = \ln(6rk^{i})/\eps$ choices is at most
\[ k^{i} \cdot (1-\eps)^{m_i} \leq k^{i} \cdot e^{-m_i \cdot \eps} \leq \tfrac{1}{6r}.\]
To conclude this process, let $m' = \sum_{i=0}^{r-1}{m_i} \leq r \cdot \frac{\ln (6rk^{r-1})}{\eps}$, and apply the union bound over the above $r$ iterations to obtain that throughout the first $m'$ choices of random sets from $\binom{[n]}{k}$, the probability to pick all the desired sets described above is at least $1-r \cdot \frac{1}{6r} = 5/6$. We let $\calR$ denote the collection of those sets.

We next show, for $\eps \geq 2 \cdot (\frac{k^2}{n})^r$, that the number of sets in $\calF_1$ that are disjoint from at least one set of $\calR$ is at least $\frac{\eps}{2} \cdot \binom{n}{k}$.
Indeed, by $B \cap C = \emptyset$, the sets of $\calF_1$ do not intersect the sets of $\calF_2$ at the elements of $A$.
Therefore, every set of $\calF_1$ that intersects all the sets of $\calR$ must include some element $j_1 \in F \setminus A$, some element $j_2 \in F_{j_1} \setminus A$, some element $j_3 \in F_{j_1,j_2} \setminus A$, and so on, up to an element $j_r \in F_{j_1, \ldots, j_{r-1}} \setminus A$. By the definition of those sets, the $r$ elements $j_1, \ldots, j_r$ are distinct. It therefore follows that there exists a collection of at most $k^r$ subsets of $[n] \setminus A$ of size $r$, such that every set of $\calF_1$ that intersects all the sets of $\calR$ contains at least one set of the collection. This implies that the number of those sets of $\calF_1$ does not exceed $k^r \cdot \binom{n-r}{k-r}$.
Since the pair $(\calF_1,\calF_2)$ is $\eps$-far from cross-intersecting, it follows that $|\calF_1| > \eps \cdot \binom{n}{k}$.
This implies that the number of sets in $\calF_1$ that are disjoint from at least one set of $\calR$ is at least
\[ |\calF_1| - k^r \cdot \binom{n-r}{k-r} > \eps \cdot \binom{n}{k} - k^r \cdot \binom{n-r}{k-r} \geq \eps \cdot \binom{n}{k} - k^r \cdot \bigg (\frac{k}{n} \bigg )^r \cdot \binom{n}{k} \geq \frac{\eps}{2} \cdot \binom{n}{k}, \]
where the last inequality relies on the assumption $\eps \geq 2 \cdot (\frac{k^2}{n})^r$.

Now, conditioned on the event that the sets chosen so far include the sets of a collection $\calR$ as above, we consider the choice of additional $4/\eps$ random sets from $\binom{[n]}{k}$.
The discussion from the previous paragraph yields that a random set chosen uniformly from $\binom{[n]}{k}$ lies in $\calF_1$ and is disjoint from at least one set of $\calR$ with probability at least $\eps/2$.
This implies that the probability that no such set is picked throughout the next $4/\eps$ choices is at most $(1-\eps/2)^{4/\eps} \leq e^{-2} < 1/6$.

Finally, let $m = m' +\frac{4}{\eps}$, and notice that $m \leq c \cdot \frac{\ln k}{\eps}$ for some $c = c(r)$.
Apply again the union bound to obtain that the probability that $m$ random sets, chosen uniformly and independently from $\binom{[n]}{k}$, do not include a collection $\calR$ as above and a set in $\calF_1$ that is disjoint from at least one set of $\calR$ is at most $1/3$. Recalling that the sets of $\calR$ lie in $\calF_2$ and that $\calF_1,\calF_2 \subseteq \calF$, it follows that with probability at least $2/3$, the $m$ random sets include two disjoint sets that lie in $\calF$. This completes the proof.
\end{proof}

We are ready to prove Theorem~\ref{thm:general_r}.

\begin{proof}[ of Theorem~\ref{thm:general_r}]
For integers $n,k,m$ and for a real $\eps \in [0,1)$, consider the $\textsc{Canonical}$ $\textsc{Tester}~(n,k,m)$ for the $\Inter_{n,k,\eps}$ problem.
As previously mentioned, this tester is non-adaptive and one-sided error.
Fix an integer $r \geq 1$, let $c = c(r)$ be the constant given in Lemma~\ref{lemma:samples}, and suppose that
\begin{eqnarray}\label{eq:eps,m}
\eps \geq 2 \cdot 3^{r^2} \cdot \bigg (\frac{k^2}{n} \bigg )^r~~~\mbox{and}~~~m \geq c \cdot 3^{r^2} \cdot \frac{\ln k}{\eps}.
\end{eqnarray}
It suffices to prove that if a function $f: \binom{[n]}{k} \rightarrow \{0,1\}$ is $\eps$-far from intersecting, then the $m$ sets picked by our tester include with probability at least $2/3$ two disjoint sets that lie in $f^{-1}(1)$. Indeed, this implies that the tester rejects such an $f$ with the desired probability.

Let $f: \binom{[n]}{k} \rightarrow \{0,1\}$ be a function $\eps$-far from intersecting. Consider the family $\calF = f^{-1}(1)$, and note that $\calF$ is $\eps$-far from intersecting, hence the pair $(\calF,\calF)$ is $\eps$-far from cross-intersecting.
Apply Lemma~\ref{lemma:ABCinduction} with $\calF_1 = \calF_2 = \calF$ and with $t=r$, to obtain that there exist sets $A \subseteq [n]$ and $B, C \subseteq A$ with $B \cap C = \emptyset$, such that for $\tilde{\eps} = \frac{\eps}{3^{r^2}}$, it holds that
\begin{enumerate}
  \item the pair $(\calF(A_{\downarrow B}),\calF(A_{\downarrow C}))$ is $\tilde{\eps}$-far from cross-intersecting, and
  \item\label{itm:222} either there is no subset of $[n] \setminus A$ of size at most $r$ that $\tilde{\eps}$-captures $\calF_2(A_{\downarrow C})$, or $|C| \geq r$.
\end{enumerate}
Moreover, our assumption on $\eps$ in~\eqref{eq:eps,m} clearly implies that $\eps \geq 3^{r^2} \cdot (\frac{k}{n})^{r}$, hence it follows from the lemma that $|C|<r$. We thus derive from Item~\ref{itm:222} above that there is no subset of $[n] \setminus A$ of size at most $r$ that $\tilde{\eps}$-captures $\calF(A_{\downarrow C})$. This obviously implies that there is no subset of $[n] \setminus A$ of size at most $r-1$ that $\tilde{\eps}$-captures $\calF(A_{\downarrow C})$.

Now, by our assumptions in~\eqref{eq:eps,m}, it holds that $\tilde{\eps} \geq 2 \cdot (\frac{k^2}{n})^r$ and $m \geq c \cdot \frac{\ln k}{\tilde{\eps}}$. Therefore, we can apply Lemma~\ref{lemma:samples} with the family $\calF$, the sets $A,B,C$, the integer $r$, and the real $\tilde{\eps}$, to obtain that with probability at least $2/3$, the $m$ sets picked by our tester include two disjoint sets that lie in $\calF$, as required.

Finally, let $m$ be the smallest integer satisfying the condition in~\eqref{eq:eps,m}.
Observe that the running time of our tester is polynomial in $n$ and $m$, hence it is polynomial in $n$ and $1/\eps$.
By our assumption on $\eps$ in~\eqref{eq:eps,m}, the running time is polynomial in $n$, and the proof is completed.
\end{proof}

\subsection{The case \texorpdfstring{$n = \Theta(k)$}{n=Theta(k)}}\label{sec:n=alpha*k}

We turn our attention now to the $\Inter_{n,k,\eps}$ problem, where the integers $n$ and $k$ satisfy $n = \alpha \cdot k$ for an arbitrary constant $\alpha \geq 2$.
We first observe that for this range of parameters, the canonical tester is not effective, even for a constant $\eps$.
To see this, observe that two random sets, chosen uniformly and independently from $\binom{[n]}{k}$, are disjoint with probability $\binom{n-k}{k} / \binom{n}{k}$.
For $n = \alpha \cdot k$, with $\alpha \geq 2$ being a constant, this probability decreases exponentially in $n$. Therefore, a collection of random sets from $\binom{[n]}{k}$ is unlikely to include even a single pair of disjoint sets, unless their number grows exponentially in $n$. As a result, the canonical tester does not provide a useful upper bound on the query complexity of the $\Inter_{n,k,\eps}$ problem in this setting.

To overcome this difficulty, we consider a slightly different tester for the $\Inter_{n,k,\eps}$ problem, which picks random pairs of disjoint sets from $\binom{[n]}{k}$ and checks if at least one of the pairs demonstrates a violation of intersectingness for the tested function. This tester allows us to prove the following result.

\begin{theorem}\label{thm:n=alpha*k}
For all reals $\alpha \geq 2$ and $\eps \in (0,1)$, there exists some $c = c(\alpha,\eps)$, such that for all sufficiently large integers $n$ and $k$ with $n = \alpha \cdot k$, there exists a non-adaptive one-sided error tester for $\Inter_{n,k,\eps}$ with $c$ queries and success probability at least $2/3$.
\end{theorem}

Theorem~\ref{thm:n=alpha*k} is obtained using the following result, that was proved (in a generalized form) by Friedgut and Regev~\cite{FR18}.
\begin{theorem}[\cite{FR18}]\label{thm:FR}
For all reals $\alpha > 2$ and $\eps \in (0,1)$, there exist a real $\delta = \delta(\alpha,\eps) > 0$ and an integer $j = j(\alpha,\eps)$, such that for all sufficiently large integers $n$ and $k$ with $n = \alpha \cdot k$, the following holds.
Suppose that $\calF \subseteq \binom{[n]}{k}$ is a family, such that for every intersecting $j$-junta $\calJ$ over $[n]$, it holds that $|\calF \setminus \calJ| > \eps \cdot \binom{n}{k}$.
Then, a random unordered pair $\{A,B\}$ of two disjoint sets of $\binom{[n]}{k}$, chosen uniformly from all such pairs, satisfies $A \in \calF$ and $B \in \calF$ with probability at least $\delta$.
\end{theorem}

\begin{proof}[ of Theorem~\ref{thm:n=alpha*k}]
Fix $\alpha \geq 2$ and $\eps \in (0,1)$, and let $n$ and $k$ be sufficiently large integers with $n = \alpha \cdot k$.
For an integer $m$, consider the following tester for the $\Inter_{n,k,\eps}$ problem.
\begin{tcolorbox}
Input: Query access to a function $f: \binom{[n]}{k} \rightarrow \{0,1\}$.
\begin{enumerate}
  \item Pick $m$ unordered pairs $\{A_1, B_1\}, \ldots, \{A_m,B_m\}$ of sets of $\binom{[n]}{k}$ with $A_i \cap B_i = \emptyset$ for all $i \in [m]$ uniformly and independently at random.
  \item Query $f$ for the values of $f(A_i)$ and $f(B_i)$ for each $i \in [m]$.
  \item If for each $i \in [m]$, it holds that $f(A_i)=0$ or $f(B_i)=0$, then accept, and otherwise reject.
\end{enumerate}
\end{tcolorbox}

The above tester is clearly non-adaptive.
Since it accepts every intersecting function with probability $1$, its error is one-sided.
It suffices to show that there exists $m = m(\alpha,\eps)$, such that if a function $f: \binom{[n]}{k} \rightarrow \{0,1\}$ is $\eps$-far from intersecting, then our tester rejects it with probability at least $2/3$. Let $f: \binom{[n]}{k} \rightarrow \{0,1\}$ be a function $\eps$-far from intersecting.

Consider first the simple case of $\alpha = 2$. Here, the fact that $f$ is $\eps$-far from intersecting implies that more than $\eps \cdot \binom{n}{k}$ of the $\frac{1}{2} \cdot \binom{n}{k}$ unordered pairs of disjoint sets of $\binom{[n]}{k}$ violate the intersectingness. For $m = 1/\eps$, the probability that our tester accepts $f$ is at most $(1-2\eps)^{m} = (1-2\eps)^{1/\eps} \leq e^{-2}$. Hence, with the complement probability, which exceeds $2/3$, our tester rejects $f$, as desired.

Next, suppose that $\alpha>2$, let $\delta = \delta(\alpha,\eps) > 0$ and $j = j(\alpha,\eps)$ be the constants given in Theorem~\ref{thm:FR}, and put $\calF = f^{-1}(1)$. Since $\calF$ is $\eps$-far from intersecting, it follows that every intersecting $j$-junta $\calJ$ over $[n]$ satisfies $|\calF \setminus \calJ| > \eps \cdot \binom{n}{k}$.
By Theorem~\ref{thm:FR}, at least $\delta$ fraction of the unordered pairs of disjoint sets of $\binom{[n]}{k}$ violate the intersectingness.
Letting $m = 2/\delta$, the probability that our tester accepts $f$ is at most $(1-\delta)^m = (1-\delta)^{2/\delta} \leq e^{-2}$. As before, it follows that with probability at least $2/3$, our tester rejects $f$, and we are done.
\end{proof}

\section{Lower Bound}\label{sec:lower}

In this section, we prove our lower bound on the query complexity of the $\Inter_{n,k,\eps}$ problem and confirm Theorem~\ref{thm:LBIntro}.
As mentioned earlier, the result can also be derived from~\cite{Fischer24}.

Before turning to the proof, let us mention that for integers $n$ and $k$ with $n \geq 2k$, the Kneser graph $K(n,k)$ is the graph on the vertex set $\binom{[n]}{k}$, where two sets are adjacent if and only if they are disjoint.
We observe that if the number of vertices $\binom{n}{k}$ is even, $K(n,k)$ admits a perfect matching, and that otherwise, it admits a matching that misses a single vertex. Indeed, for $n=2k$, the collection of all edges of $K(n,k)$ forms a perfect matching, and for $(n,k)=(5,2)$, the graph $K(n,k)$ coincides with the Petersen graph that has a perfect matching as well. For all other admissible values of $n$ and $k$, a recent result of Merino, M{\"{u}}tze, and Namrata~\cite{MerinoMN23} asserts that $K(n,k)$ admits a Hamilton cycle, yielding the existence of the desired matching. Equipped with this observation, we state and prove the following result.

\begin{theorem}\label{thm:LB}
There exists a constant $c > 0$, such that for all integers $n$ and $k$ with $n \geq 2k$ and for any real $\eps = \eps(n,k)$ with $\binom{n}{k}^{-1} \leq \eps < \frac{1}{2}$, the query complexity of every (possibly adaptive, two-sided error) tester for $\Inter_{n,k,\eps}$ with success probability at least $2/3$ is at least $c \cdot \frac{1}{\eps}$.
\end{theorem}

\begin{proof}
Let $n$ and $k$ be integers with $n \geq 2k$, and suppose that $\binom{n}{k}^{-1} \leq \eps < \frac{1}{2}$.
We assume throughout the proof that the Binomial coefficient $\binom{n}{k}$ is even. A similar argument can be applied when it is odd.
Under this assumption, the discussion that precedes the statement of Theorem~\ref{thm:LB} yields that there exists a partition $\calP$ of $\binom{[n]}{k}$ into pairs of disjoint sets. Note that $|\calP| = \frac{1}{2} \cdot \binom{n}{k}$.

We define a distribution $\calD_{no}$ over the Boolean functions on $\binom{[n]}{k}$ as follows.
Let $N$ denote the smallest integer satisfying $N> \eps \cdot \binom{n}{k}$, and notice that the assumption $\eps < \frac{1}{2}$ implies that $N \leq |\calP|$.
To draw a function $f: \binom{[n]}{k} \rightarrow \{0,1\}$ from $\calD_{no}$, we pick uniformly at random a collection $\calP' \subseteq \calP$ of size $N$, and for each set $A \in \binom{[n]}{k}$, we define $f(A)=1$ if $A$ lies in the pairs of $\calP'$, and $f(A)=0$ otherwise.
Observe that any function $f$ in the support of $\calD_{no}$ is $\eps$-far from intersecting, because to make it intersecting, one has to change its value on at least one set from each of the $N > \eps \cdot \binom{n}{k}$ pairs of $\calP'$, and those pairs are pairwise disjoint.
Note that the fact that $\calP$ forms a partition of $\binom{[n]}{k}$ implies that a function $f$ drawn from $\calD_{no}$ satisfies for every set $A \in \binom{[n]}{k}$ that
\begin{eqnarray}\label{eq:Prob_f(A)=1}
\Prob{}{f(A)=1} =  \frac{N}{|\calP|} \leq \frac{\eps \cdot \binom{n}{k}+1}{\frac{1}{2} \cdot \binom{n}{k}} = 2\eps + \frac{2}{\binom{n}{k}} \leq 4 \eps,
\end{eqnarray}
where the second inequality follows from the assumption $\eps \geq \binom{n}{k}^{-1}$.

Now, suppose that there exists a (two-sided error) tester $T$ for $\Inter_{n,k,\eps}$ with success probability at least $2/3$ and query complexity $q$.
Consider the distribution $\calD$ over the Boolean functions on $\binom{[n]}{k}$, which with probability $1/2$ returns the constant zero function and with probability $1/2$ returns a function drawn from $\calD_{no}$.
By Yao's minimax principle~\cite{Yao77}, there exists a choice for the random coins of $T$ that results in a deterministic tester $T'$ with query complexity at most $q$, whose success probability on a random function drawn from $\calD$ is at least $2/3$.
Let $A_1, \ldots, A_{q_0} \in \binom{[n]}{k}$ denote the queries of $T'$ while running on the constant zero function, and note that $q_0 \leq q$.
We define a tester $T''$ that queries an input function on the sets $A_1, \ldots, A_{q_0}$, and then, if at least one of the answers is $1$ rejects the function, and otherwise, acts according to $T'$.
Recalling that any nonzero function in the support of $\calD$ is $\eps$-far from intersecting, it follows that the success probability of $T''$ on a random function drawn from $\calD$ is at least $2/3$.
Since the latter probability exceeds $1/2$, we derive that $T''$ accepts the intersecting constant zero function, and consequently, any other function $f$ satisfying $f(A_i)=0$ for all $i \in [q_0]$.

Finally, to obtain a lower bound on $q$, let us examine how the tester $T''$ acts on a function drawn from $\calD_{no}$.
Combining~\eqref{eq:Prob_f(A)=1} with the union bound yields that a function $f$ drawn from $\calD_{no}$ satisfies $f(A_i)=1$ for some $i \in [q_0]$ with probability at most $q_0 \cdot 4\eps$.
Therefore, $T''$ accepts a function drawn from $\calD_{no}$ with probability at least $1-q_0 \cdot 4 \eps$.
Since the functions of the support of $\calD_{no}$ are $\eps$-far from intersecting, it follows that $T''$ errs on such a function with probability at least $1-q_0 \cdot 4 \eps$.
On the other hand, since $T''$  errs on a function drawn from $\calD$ with probability at most $1/3$, it errs on a function drawn from $\calD_{no}$ with probability at most $\frac{1/3}{1/2} = 2/3$.
This implies that $1 - q_0 \cdot 4 \eps \leq 2/3$, which yields that $q \geq q_0 \geq c/\eps$ for a constant $c>0$, as desired.
\end{proof}

\begin{remark}
We note that, for all integers $n$ and $k$ with $n \geq 2k$, the condition $\eps < \frac{1}{2}$ given in Theorem~\ref{thm:LB} makes the $\Inter_{n,k,\eps}$ problem non-trivial, in the sense that there exists a family of sets from $\binom{[n]}{k}$ that is $\eps$-far from intersecting.
Indeed, the Erd{\H{o}}s--Ko--Rado theorem~\cite{EKR61} implies that to make the family $\binom{[n]}{k}$ of all $k$-subsets of $[n]$ intersecting, one has to remove at least $(1-\frac{k}{n}) \cdot \binom{n}{k} \geq \frac{1}{2} \cdot \binom{n}{k} > \eps \cdot \binom{n}{k}$ of its sets.
\end{remark}

\section*{Acknowledgments}
We would like to thank the anonymous reviewers for their useful comments.

\bibliographystyle{abbrv}
\bibliography{intersecting}

\end{document}